\documentclass[11pt]{article}
\usepackage{geometry}

\usepackage{graphicx} %
\usepackage{amsfonts}
\usepackage{amsmath}
\usepackage{amssymb}
\usepackage{amsthm}
\usepackage{mathtools}
\usepackage{xcolor} %
\usepackage{natbib}
\usepackage{comment}
\usepackage{csquotes}
\usepackage{authblk} %
\usepackage{todonotes}

\usepackage[colorlinks=true,
            linkcolor=violet,
            citecolor=violet,
            urlcolor=violet]{hyperref}

\newtheorem{theorem}{Theorem}[section]
\newtheorem{defn}[theorem]{Definition}

\newtheorem{rem}[theorem]{Remark}
\newtheorem{lemma}[theorem]{Lemma}
\newtheorem{claim}[theorem]{Claim}

\newcommand{\Dgm}{\operatorname{Dgm}}
\newcommand{\ECC}{\operatorname{ECC}}

\newcommand\labelAndRemember[2]{%
  \expandafter\gdef\csname labeled:#1\endcsname{#2}%
  \label{#1}#2}
\newcommand\recallLabel[1]{%
  \csname labeled:#1\endcsname\tag{\ref{#1}}}

\newcommand{\scalarproduct}[2]{\langle #1, #2 \rangle}

\newcommand{\twiddle}[1]{\widetilde{#1}}
\newcommand{\inv}[0]{^{-1}}
\newcommand{\lra}[0]{\longrightarrow}

\newcommand{\R}{\mathbb{R}}
\renewcommand{\S}{\mathbb{S}}

\newcommand{\Z}{\mathbb{Z}}
\renewcommand{\phi}{\varphi} 
\newcommand{\SELECT}{\operatorname{SELECT}}

\newcommand{\ReFLECT}{\operatorname{ReFLECT}}
\newcommand{\ECT}{\operatorname{ECT}}
\newcommand{\LECT}{\operatorname{LECT}}
\newcommand{\PL}{\operatorname{PL}}
\newcommand{\Func}{\operatorname{Func}}

\newcommand{\mirror}[1]{\reflectbox{$#1$}}

\title{\LARGE On the Stability of the Euler Characteristic Transform for a Perturbed Embedding}
\author[1]{\large Jasmine George}
\author[2]{\large Oscar Lledo Osborn}
\author[3, 4]{\large Elizabeth Munch}
\author[5]{\large Messiah Ridgley II}
\author[6]{\large Elena Xinyi Wang}

\affil[1]{\footnotesize Department of Mathematics, University of Dayton}
\affil[2]{\footnotesize Department of Mathematics, Hamilton College}
\affil[3]{\footnotesize Department of Computational Mathematics, Science, and Engineering, Michigan State University}
\affil[4]{\footnotesize Department of Mathematics, Michigan State University}
\affil[5]{\footnotesize Department of Mathematics, Brandeis University}
\affil[6]{\footnotesize Institute of Geometry, Graz University of Technology} %

\begin{document}
\date{} %
\maketitle

\begin{abstract}
The Euler Characteristic Transform (ECT) is a robust method for shape classification.
It takes an embedded shape and, for each direction, computes a piecewise constant function representing the Euler Characteristic of the shape's sublevel sets, which are defined by the height function in that direction.
It has applications in TDA inverse problems, such as shape reconstruction, and is also employed with machine learning methodologies.
In this paper, we define a distance between the ECTs of two distinct geometric embeddings of the same abstract simplicial complex and provide an upper bound for this distance.
The Super Lifted Euler Characteristic Transform (SELECT), a related construction, extends the ECT to scalar fields defined on shapes.
We establish a similar distance bound for SELECT, specifically when applied to fields defined on embedded simplicial complexes.
\end{abstract}

\section{Introduction}

Shape classification using Topological Data Analysis (TDA) has found diverse applications, including the analysis of barley seeds~\cite{amezquita2022measuring}, proteins~\cite{tang2022topological}, and primate molars~\cite{wang2021statistical}.
When data intrinsically encodes \enquote{shape,} TDA provides prudent analytical approaches.
TDA tools are generally defined for \enquote{tame} sets, typically elements of an $o$-minimal structure~\cite{turner2014persistent}.
The Triangulation Theorem (Theorem~\ref{thm:triangulation}) states that any tame set into a simplicial complex. 
Simplicial complexes, as combinatorial objects, offer advantages for both theoretical analysis and computation.

Within TDA, various tools grounded in algebraic topology have been developed to analyze data while preserving crucial shape information; persistent homology is among the most powerful~\cite{EdelsbrunnerPH}.
A cornerstone of persistent homology is its stability theorem, guaranteeing that its standard visualization, the persistence diagram, remains robust against small data perturbations~\cite{CohenSteiner2007Stability}.
The later-developed Persistent Homology Transform (PHT) offers a robust and injective descriptor by capturing a shape's persistent homology from all viewing directions~\cite{turner2014persistent}.
However, despite its theoretical strengths, the PHT is computationally demanding, with the number of required directions growing exponentially with data dimensionality~\cite{curry2022many}.

Consequently, this work focuses on the more computationally tractable Euler Characteristic Curves (ECC) and the Euler Characteristic Transform (ECT)~\cite{turner2014persistent}, along with their recent variants~\cite{kirveslahti2024representing}.
For ECC and its variants, stability properties and methods for statistical inference are well-explored~\cite{Hacquard2024ECP, miller2025filtered, perez2022euler, roycraft2023boostrap, fasy2014confidence, dlotko2023euler}.
The ECT also boasts a solid theoretical foundation~\cite{curry2022many, turner2014persistent}, underpinning its successful application in diverse fields~\cite{amezquita2022measuring, CrawfordMonod, Marsh2024Temporal, Nadimpalli2023ECT, wang2021statistical}.

A key limitation of ECT, however, is the lack of general stability results, a topic that has garnered recent attention.
Notably, Nadimpalli et al.~\cite{Nadimpalli2023ECT} established stability for an ECT variant on binary images, with bounds dependent on pixel count.
Separately, Marsh et al.~\cite{marsh2023stabilityic} introduced a novel curvature-sensitive metric for one-dimensional shapes, proving ECT stability with respect to this metric, independent of the shape's triangulation.

In this work, we address ECT stability by considering a more restricted setting: comparing shapes that share an underlying simplicial complex.
Within this context, Skraba and Turner~\cite{skraba2020wasserstein} previously proved the stability of the PHT using Wasserstein distances.
Following their analysis, we prove the stability of the ECT.
Our main contributions are:
\begin{enumerate}
    \item Demonstrating the stability of the $\ECT$ with respect to the Wasserstein distance between distinct geometric embeddings of the same abstract simplicial complex,
    \item Establishing the stability of $\SELECT$ when applied to functions (filtrations) defined on embedded simplicial complexes, with respect to changes in these embeddings.
\end{enumerate}

\section{Preliminaries}
In this section, we cover the prerequisite definitions and give a high-level overview of persistent homology. 
Refer to~\cite{curry2022many},~\cite{dey2022computational} and~\cite{turner2014persistent} for details.

\subsection{Simplicies and Simplicial Complexes}

 \begin{defn}\label{simplex}
 A \textbf{geometric $k$-simplex} is the convex hull of $k+1$ affinely independent points (vertices) $\{ v_0, \cdots, v_k \}$ in $\R^d$, 
    $$\{x \in \R^d |\text{ there exists }\lambda_0,\cdots \lambda_k \in \R,  \sum_{i=0}^{k} \lambda_i = 1, x = \sum_{i=0}^{k} \lambda_i v_i  \}.$$
    Additionally, $|\sigma|$ denotes the dimension of a simplex. 

     A simplex $\sigma'$ such that $\sigma' \subseteq \sigma$ is called a \textbf{face} of $\sigma$, e.g., an edge of a triangle is a face of the triangle, and a vertex is a face of both an edge and a triangle. 
 \end{defn}

\begin{figure}
    \centering
    \includegraphics[width=0.75\linewidth]{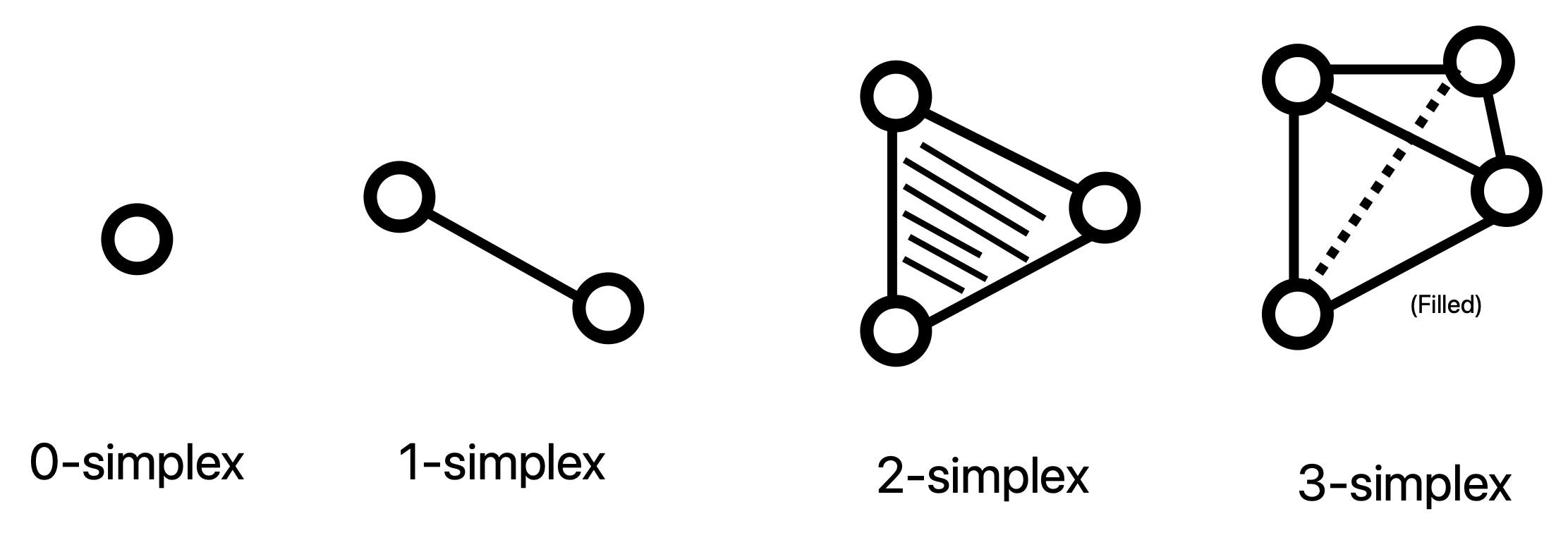}
    \caption{Examples of geometric $k$-simplices}
    \label{fig:k-simplex}
\end{figure}
See Figure \ref{fig:k-simplex} for examples of simplices. 
Simplices are common building blocks in algebraic topology. 
One might think of each dimensional simplex as the \enquote{building block} or \enquote{prototype} of manifolds of that dimension, i.e., triangles are the building blocks of two-dimensional manifolds, tetrahedra the building blocks of three-dimensional manifolds, and so on. 

A \textit{geometric simplicial complex} is a collection of geometric simplices such that, for any simplex $\sigma \in K$, all of its faces $\sigma' \subseteq \sigma$ are also in the complex ($\sigma' \in K$), and the intersection of any two simplices is either empty or a common face. 
We denote the vertex set of a geometric simplicial complex $K$ by $V(K)$, consisting of all the 0-simplices in the complex.
A \textit{$k$-complex} is a simplicial complex where the maximum dimension of any simplex it contains is $k$.
An example of a geometric simplicial complex is given on the left side of Figure \ref{fig:DT_function}. 
It is a 2-complex as its highest dimensional simplex is a triangle (2-simplex).

An important result in topology is that any nice enough \enquote{shape} can be represented as a simplicial complex~\cite{van1998tame}. 

\begin{theorem}[Triangulation Theorem~\cite{van1998tame}]\label{thm:triangulation}
Any tame set admits a definable bijection with a subcollection of open simplices in the geometric realization of a finite Euclidean simplicial complex. Moreover, this bijection can be made to respect a partition of a tame set into tame subsets.
\end{theorem}

While a full understanding of o-minimal topology is outside of the scope of this paper, we can use this result to replace all of our structures of interest with simplicial complexes.
For a simplicial complex to triangulate (represent) a shape, it must be homeomorphic to the shape. 
In the context of this paper, while many of the tools we will discuss are well defined for general topological surfaces, we restrict ourselves to this simplicial representation.

Simplicial complexes can be viewed as a generalization of a graph, and as such carry combinatorial and algebraic properties. 
Despite this similarity, simplicial complexes need not exist geometrically, and can be defined abstractly.
We now introduce the combinatorial abstraction of geometric complexes.

\begin{defn}\label{Abstract Simplicial Complex}  Let $V$ be a set. A collection $K$ of non-empty subsets of $V$ is called an \textbf{abstract simplicial complex} if for any simplex $\sigma \in K$, every non-empty $\sigma ' \subseteq \sigma$ is also in $K$. Elements of $V$ are called vertices, and a simplex $\sigma$ is called a $k$-simplex if it has $k+1$ vertices. A complex is called a $k$-complex if $\underset{\sigma\in K}{\text{max}}(|\sigma|)=k$.
 \end{defn}
When it becomes necessary to distinguish between the vertex sets of different complexes, we will use the notation $V(K)$ to denote the vertex set of an abstract complex $K$, as we do for geometric complexes.

By ignoring the geometric information, the collection of simplices in a geometric simplicial complex forms an abstract simplicial complex.
To go the other way, a geometric complex $\mathcal{K}$ is called a \textit{geometric realization} of an abstract simplicial complex $K$ if there exists an embedding $e:V(K) \lra \R^n$ that takes every $k$-simplex $\{v_0,\cdots,v_k\}$ to a $k$-simplex in $J$ that is the convex hull of $\{e(v_0),\cdots,e(v_k)\}$ and the images of the vertices $\{e(v_0),\cdots,e(v_k)\}$ are affinely independent.

Any abstract $k$-complex can be embedded geometrically (without self-intersections) in $\R^{2k+1}$, and any abstract complex on $n$ vertices can be embedded in $\R^{n-1}$ as a subcomplex of an $(n-1)$-simplex (i.e. three triangles on four points as a subcomplex of a tetrahedron).
 For the remainder of the paper, unless specified to be abstract, simplicial complexes can be assumed to be geometric.
 Additionally, we assume all of them to be finite.

 \subsection{The Euler Characteristic, Curve, and Transform}
 
The Euler characteristic is a topological invariant.
In the case of simplicial complexes, it has the following formulation.
\begin{defn}\label{Euler Characteristic} 
The \textbf{Euler Characteristic} of a simplicial complex $K$ is the alternating sum of the number of simplices of each dimension. 
\begin{align*}
    \chi(K)&:=\underset{\sigma\in K}{\sum}(-1)^{|\sigma|}
\end{align*}
Alternatively, this can be written as 
\begin{align*}
    \chi(K)&=\underset{d=0}{\sum}(-1)^d \cdot |\{\sigma \in K \big{|}\;|\sigma|=d\}|\\ 
    &=|\text{vertices}|-|\text{edges}|+|\text{faces}|-|\text{tetrahedra}|\cdots \\
\end{align*}

\end{defn}

Two shapes are said to be \textit{homotopy equivalent} if they can be continuously deformed into each other. The Euler Characteristic is a homotopy invariant, and thus assigns the same value to any two shapes that are homotopy equivalent. 
However, many non-homotopic shapes also have the same Euler characteristic. 
To determine a shape definitively, we need to add more information to the Euler Characteristic.
A step in this direction is the Euler Characteristic Curve, which relies on the height function for a given direction, $\nu\in \mathbb{S}^{d-1}$.
\begin{equation*}
\begin{matrix}
   h_\nu : & K & \longrightarrow & \R \\
    & \sigma  & \longmapsto & \underset{v\in V(\sigma)}{\max} \langle v, \nu \rangle 
\end{matrix}
\end{equation*}
Where $\langle\cdot,\nu\rangle$ is the dot product for $\R^d$ and $V(\sigma)$ is the vertex set of each simplex. 
The height function is defined so that a simplex will never appear without its boundary, ensuring each sublevel set of the height function in a given direction is a subcomplex of our larger complex.  
These sublevel sets are given by $K_{\nu,a}=\{\sigma\in K|\; h_\nu(\sigma)\leq a\}$ for $\nu \in \mathbb{S}^{d-1},\;a \in \R$.
\begin{defn}\label{ECC}

The \textbf{Euler Characteristic Curve} $(\ECC)$ of a finite simplicial complex $K\subset\R$ for a given vector $\nu\in \mathbb{S}^{d-1}$ is given by:$$\ECC_\nu(K,a):=\chi(K_{\nu,a})$$
We denote the function as a whole by $\ECC_\nu(K): \R \to \mathbb{Z}$. 
\end{defn}
The $\ECC$ returns the Euler Characteristic for each sublevel set of the height function in a given direction, endowing it with more information than just the Euler Characteristic. 
Intuitively, one should think of this as a scan in a direction, as in Figure \ref{fig:Leaf ECC}. 
\begin{figure}
    \centering
    \includegraphics[width = 0.25\textwidth]{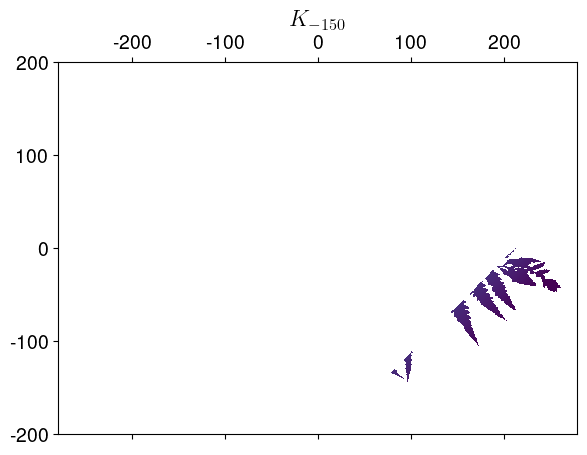}
    \includegraphics[width = 0.25\textwidth]{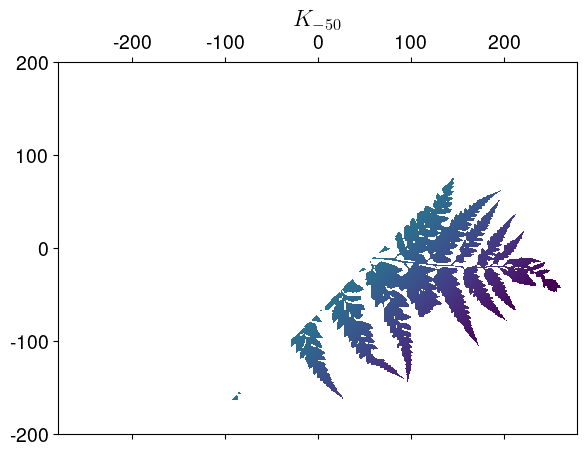}
    \includegraphics[width = 0.25\textwidth]{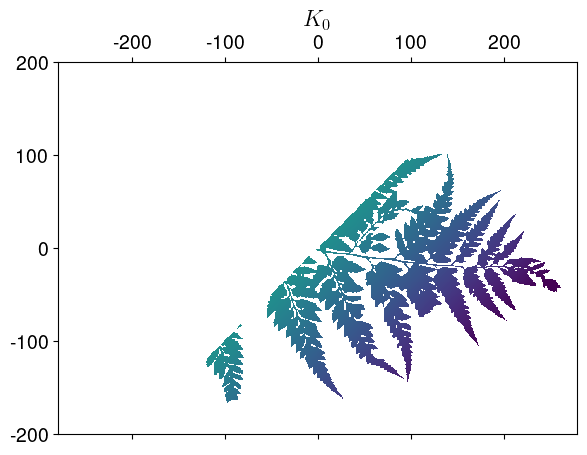}
    
    \includegraphics[width = 0.25\textwidth]{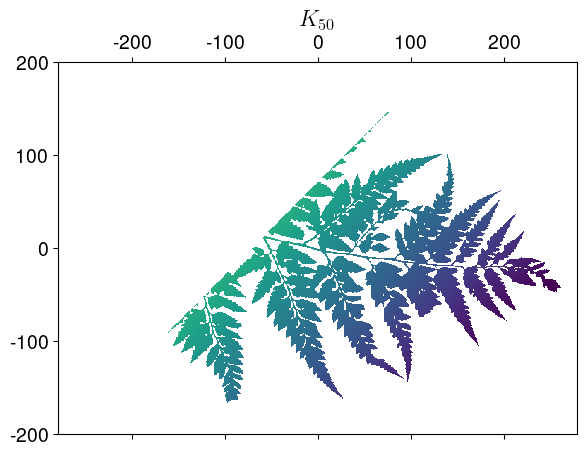}
    \includegraphics[width = 0.25\textwidth ]{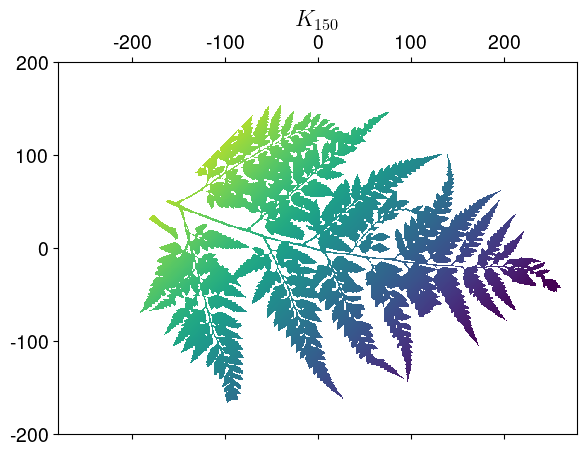}
    \includegraphics[width = 0.25\textwidth]{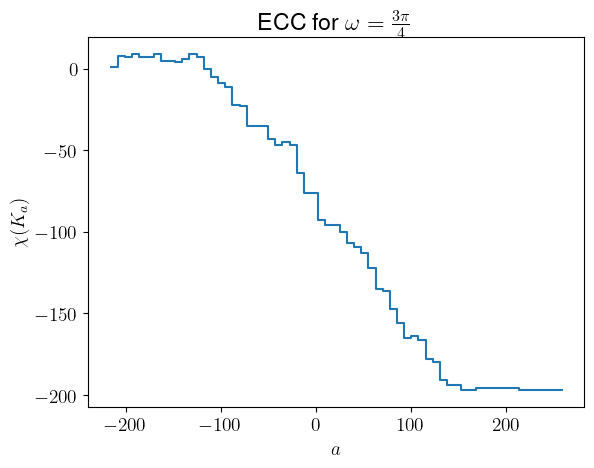}
    \caption{A ``scan'' of a fern leaf along the $\frac{3\pi}{4}$ direction and the associated ECC. Figure from~\cite{munch2023invitation}. }
    \label{fig:Leaf ECC}
\end{figure}

Two different directions might produce drastically different Euler Characteristic Curves, but small perturbations in the direction should not change the curve that much. 
In order to meaningfully talk about how different the $\ECC$s of two different shapes are, we use the $L_1$ norm.
Specifically, $\mathcal{L}^p(X)$ is the set of measurable functions $f:X \to \R$ such that
$\|f\|_p:= \left( \int_X |f|^p\right)^{1/p} <\infty$ where functions which agree almost everywhere are identified. 
We will use this distance to build distances for the inputs of interest in this paper. 

\begin{defn}\label{ECCdistance}
For simplicial complexes $K_1$ and $K_2$, scanned in directions $\nu$ and $\omega$ respectively, the $L_1$ distance between their Euler Characteristic Curves is given by
$$ 
\|\ECC_\nu(K_1)-\ECC_\omega(K_2)\|_1
:=\underset{\R}{\int}|\ECC_\nu(K_1,a)-\ECC_\omega(K_2,a)|\, da.
$$
\end{defn}

This distance is finite only if the simplicial complexes $K_1$ and $K_2$ have the same Euler Characteristic, $\chi(K_1)=\chi(K_2)$.
This condition is crucial: for a sufficiently large filtration parameter $a$, the filtered complexes $K_{1,\nu,a}$ and $K_{2,\omega,a}$ stabilize to the full complexes (i.e., $K_{1,\nu,a} = K_1$ and $K_{2,\omega,a} = K_2$).
At this point, the term $|\chi(K_1)-\chi(K_2)|$, which contributes to the distance calculation (especially for large $a$), must be zero for the distance to be finite.
If $\chi(K_1) \neq \chi(K_2)$, this term is non-zero, leading to an infinite distance.
Therefore, this distance defines a metric on the set of simplicial complexes sharing the same Euler Characteristic, and an extended metric on the set of all simplicial complexes.

In any application, since all complexes that represent data are finite, we assume that all complexes are contained in some bounding ball with radius $B$, in which case one can meaningfully talk about the distances of two complexes with different Euler Characteristics by restricting the region of integration to $[-B,B]$.

\begin{figure}
    \centering
    \includegraphics[width = \textwidth]{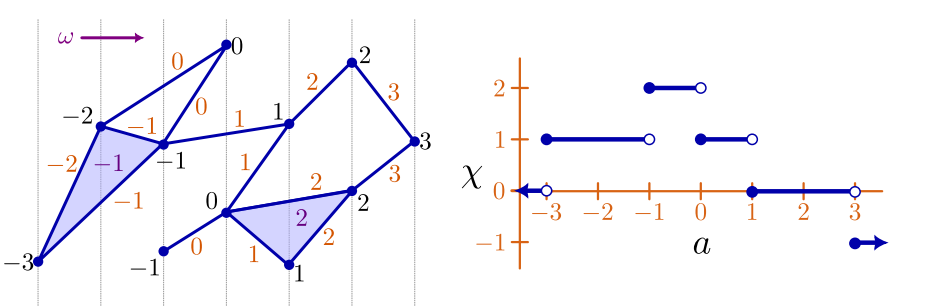}
    \caption{On the left, we have a simplicial complex with simplex-wise function $h_\omega$ labeled for the horizontal choice of $\omega = 0$. The numerical values of each vertex $v$ correspond to the value of $\langle v , \omega \rangle$.
    On the right, we have the corresponding Euler characteristic curve. Figure from \cite{munch2023invitation}.
    }
    \label{fig:DT_function}
\end{figure}
While the ECC provides more information than just the Euler Characteristic, it still cannot distinguish between shapes. 
The next step is to consider the Euler Characteristic Curve for each direction. 
\begin{defn}
\label{ECT}
Given a simplicial complex $K$ embedded in $\R^d$, the \textbf{Euler Characteristic Transform (ECT)} is 
\begin{equation*}
\begin{matrix}
   \ECT(K): & \mathbb{S}^{d-1} & \longrightarrow & \operatorname{Func}(\R,\mathbb{Z}) \\
    & \nu & \longmapsto & \ECC_\nu
\end{matrix}
\end{equation*}
where $\operatorname{Func}(\R,\mathbb{Z})$ denotes functions from $\R$ to $\mathbb{Z}$.

\end{defn}
While this notation $\ECT(K)$ does not include the parameters $a$ or $\nu$, it can be restricted to any $a\in\R$, $\nu\in \mathbb{S}^{d-1}$ as $\ECT(K,\nu,a)=\ECC_\nu(K,a)=\chi(K_{\nu,a})$. 
The $\ECT$ provides enough data to distinguish between shapes, as it is injective as shown in~\cite{turner2014persistent}.
This means that if two given geometric simplicial complexes have the same ECT, they must have been the same to begin with. 
\subsection{Persistent Homology}

Our results also leverage persistent homology, which is closely related to the ECT.
The \textit{k-th homology group} of a shape is a way of representing the $k$-dimensional holes that a shape has as a basis for a vector space\footnote{Generally speaking, the homology group is a module, but in TDA convention is to use coefficients from a field (e.g.,v $\mathbb{Z}_2$, $\mathbb{Q}$, $\R $), which gives the homology groups a vector space structure.}. 
The $k$-th Betti number, $\beta_k$, is the dimension of this vector space, but can also be thought of as the number of $k$-dimensional holes. 
Elements of a homology group (or its basis elements) are called homology classes.
\begin{rem}
    By the Euler-Poincaré formula, the Euler Characteristic can also be determined by an alternating sum of the Betti numbers, $\chi(K)=\underset{k}\sum(-1)^k\beta_k(K)$. 
    For a finite simplicial complex, this sum is finite because at a certain index $i_0$, $\beta_i = 0$ for all $i\geq i_0$.
    Notably, for $p$-complexes, $\beta_{k} = 0$ for $k>p$.
\end{rem}

Recall that the Euler Characteristic alone is not enough information to determine a shape. 
Given that the Euler Characteristic can be defined in terms of homology groups, it is not a surprise that homology groups alone are not enough to determine a shape. 
In order to encode more information into the Euler Characteristic, we used a function that returned a series of sublevel sets, each of which was included in the next. 
The generalization of this, a filtration, allows us to encode more information into homology groups in a similar way.

\begin{defn}\label{def:filtration}
    Fix a given simplicial complex, $K$, a a real-valued function, $f: K \rightarrow \R$, and a sequence of values $a_1 < a_2< \cdots a_n$. 
    A \textbf{filtration} is the collection of sublevel sets $K_{a_i}:=f^{-1}(-\infty,a_i]$. Note that for all $a\leq b$ there are inclusions $K_a\subseteq K_b$.  
    Additionally, we take $a_0=-\infty$ in order to set $\mathbb{T}_{a_0}=\emptyset$. 
    The filtration induced by $f$ is written as $\mathcal{F}_f$.
\end{defn}

In practice, the sequence of values, $a_0,a_1, \dots, a_n$ is chosen to be critical values so that the homology group of the sublevel sets changes at each step. 
The height function defined previously induces a filtration on a simplicial complex. 
We can also define more general functions on simplicial complexes by assigning each open simplex a value. 
Recalling the height function, a useful property is that a simplex never appears in the filtration without its boundary, and thus, for any filtration value, the resulting filtered object remains a simplicial complex.
A function that satisfies this condition is called simplex-wise monotone.

\begin{defn}\label{def:simplexfil}
A function, $f$, on a simplicial complex, $K$, is \textbf{simplex-wise monotone} if for all $\sigma' \subset \sigma$, the function satisfies $f(\sigma')\leq f(\sigma)$.
\end{defn}

We use persistence diagrams to represent how the homology groups of a filtration change over time.
A \textit{k dimensional persistence diagram} ($\Dgm_k(\mathcal{F})$) of a filtration, $\mathcal{F}$,  is a multiset of points in $\R^{2+}:=\{\R\cup\{\infty\}\}^2$ where the $x$ coordinate of a point is the filtration value $a_i$ where a k-dimensional homology class (``k-dimensional hole") appears. 
The $y$ coordinate of the same point is the filtration value $a_j$ where it disappears.
Conventionally, we refer to the appearance of a homology class as its birth, and the disappearance as its death.
If a homology class does not ever die, it has $y$-value $\infty$. 
We also allow multiplicity of each point and for infinitely many points on the diagonal. 
When discussing a diagram with all persistence classes of the space mapped, only $\Dgm(\mathcal{F})$ is used, though some marker to separate classes by dimension is prudent.
A persistence diagram gives a visual summary of the way homology changes through filtration values.
Points close to the diagonal indicate a smaller feature, whereas those further from the diagonal are generally considered more important to the structure of the shape as a whole. 

To construct a distance between two diagrams, we evoke the following matching scheme.

\begin{defn}
Let $\Dgm_k(\mathcal{F})$ and $\Dgm_k(\mathcal{G})$ be persistence diagrams represented by countable multisets in $\R^{2+}$. 
A \emph{matching} $\textbf{M}$ between $\Dgm_k(\mathcal{F})$ and $\Dgm_k(\mathcal{G})$ is a subset of $\{\Dgm_k(\mathcal{F})\cup \Delta\} \times \{\Dgm_k(\mathcal{G})\cup \Delta\}$ such that every element in $\Dgm_k(\mathcal{F})$ and $\Dgm_k(\mathcal{G})$ appears in exactly one pair. 
The abstract diagonal element $\Delta:=\{(x,x)\;|\;x\in\R\}$ can appear in many pairs.
\end{defn}

There are many distances available for use between persistence diagrams; however, the one that best suits our use is the \textit{Wasserstein Distance}.
This distance uses the $l_p$ distance for any points with non-infinite $y$ coordinate not on the diagonal. 
The Wasserstein distance is also called the \enquote{optimal transport} distance, and the distance between two points is often referred to as the \enquote{cost}.
We define the $l_p$ distance from any point in $\R^2$ to $\Delta$ by the perpendicular distance, 
$$\|(a,b)-\Delta\|_p=\inf_{t\in \R}\|(a,b)-(t,t)\|_p=\|(a,b)-\left(\frac{b-a}{2}, \frac{b-a}{2}\right)\|_p=2^{\frac{1-p}{p}} |b-a|$$
and $\|(a,b)-\Delta\|_\infty=\frac{b-a}{2}$.
Furthermore we say that $\|\Delta-\Delta\|_p=0$ for all $p$ and $\|(a,\infty)-(b, \infty)\|_p = |a-b|$.
We look at all possible matchings and pick the infimum of the cost as the distance. 

\begin{defn}[Wasserstein Distance]\label{def:wasserstein}
  Given two diagrams, $\Dgm_k(\mathcal{F})$ and $\Dgm_k(\mathcal{G})$, the $(p,q)$-Wasserstein distance is
  $$W_{p,q}(\Dgm_k(\mathcal{F}), \Dgm_k(\mathcal{G})) = \inf\limits_{\textbf{M}} \left(\sum\limits_{(x,y)\in \textbf{M}} \|x-y\|_q^p\right)^{\frac{1}{p}} $$
where $\textbf{M}\subset \{\Dgm_k(\mathcal{F})\cup \Delta\} \times \{\Dgm_k(\mathcal{G}\cup \Delta\}$ is a matching.  The total $(p,q)$-Wasserstein distance is defined as

$$W_{p,q}(\Dgm(\mathcal{F}), \Dgm(\mathcal{G})) = \left(\sum_{k}\left(W_{p,q}(\Dgm_k(\mathcal{F}), \Dgm_k(\mathcal{G})\right)^p\right)^{\frac{1}{p}}$$
\end{defn}
\begin{rem}
    For any two persistence diagrams with differing numbers of points with y-value infinity, any Wasserstein distance is infinity as $|(a,\infty)-(\infty,\infty)|=\infty$. 
    This is another similarity between the Euler Characteristic Curve and persistence diagrams, as any simplicial complex with different Euler Characteristics will have different amounts of homology classes that persist until infinity.
 \end{rem}
\section{Prior Work}
In this section, we discuss work related to our results. 

The below theorem uses the generalized definition of the ECC, that takes a general filtration rather than being restricted to a height function.
\begin{theorem}[Dłotko \& Gurnari~\cite{dlotko2023euler} Proposition 2] \label{lowerbound}
Given geometric simplicial complexes $K_1, K_2$, with filtrations $\mathcal{F}$ and $\mathcal{G}$ respectively, then the $L_1$ distance of the Euler Characteristic Curves between them is bounded by the Wasserstein distance between the filtrations.
\begin{equation*}
    {\|\ECC_\mathcal{F}(K_1)-\ECC_\mathcal{G}(K_2)\|_1 \leq 2W_{1,\infty} (\Dgm(\mathcal{F}(K_1)),\Dgm(\mathcal{G}(K_2))}
\end{equation*}
\end{theorem}
This result gives a relationship between the ECC of a complex with a filtration and the persistence diagrams obtained from the same filtration. 
We can further restrict the filtration to be the height function in a direction, making it suitable for the ECT.
\begin{equation*}
   {\|\ECC_\nu(f(K)) - \ECC_\nu(g(K))\|_1 \leq 2W_{1,\infty}(\Dgm(h_\nu^{f(K)}),\Dgm(h_\nu^{g(K)})).}
\end{equation*}

The next theorem uses the $W_{1,1}$ distance, as opposed to the $W_{1,\infty}$.
We also restrict the bound to only include embeddings of the same abstract simplicial complex, as that is a requirement for the next result we present.
\begin{theorem}[{Skraba \& Turner~\cite[Thm.~5.5]{skraba2020wasserstein}}] \label{upperbound}
Let $K$ be an abstract simplicial complex, and let $\omega_{d-2}$ denote the area of $ \mathbb{S}^{d-2}$. Given two embedding functions, $f,g:K\longrightarrow\R^d$, let $$C_d = 2 \omega_{d-2} \int_0^\frac{\pi}{2}\cos(\theta)\sin^{d-2}(\theta)d\theta \text{ and } C_K = \max\limits_{{v \in V(K)}} |\{ \sigma \;|\; v \in \sigma \}  |. $$
Then
\begin{equation*}
           {\int_{\nu\in S^{d-1}}W_{1,1}(\Dgm(h_\nu^{f(K)}),\Dgm(h_\nu^{g(K)})) d\nu \leq C_K C_d \sum_{v\in V(K)}\|f(v) - g(v)\|_2}.
\end{equation*}
\end{theorem}

While the two results utilize two different Wasserstein distances, they are closely related.
This difference is remedied via the following theorem.
\begin{theorem}[Turner~\cite{turner_medians_2020}]
For any two given persistence diagrams, $\Dgm(\mathcal{F})$ and $\Dgm(\mathcal{G})$,
\begin{equation*}
    W_{p,\infty}(\Dgm(\mathcal{F}),\Dgm(\mathcal{G}))\leq W_{p,p}(\Dgm(\mathcal{F}),\Dgm(\mathcal{G})) \leq 2W_{p,\infty}(\Dgm(\mathcal{F}),\Dgm(\mathcal{G}))
\end{equation*}
\end{theorem}

We want to use these two results to bound the distance between the Euler Characteristic Transforms based on the distance of two embedding functions of the same abstract complex. 
We propose the following metric on the space of ECT functions, building on the $L_1$ norm. 
\begin{defn}[ECT Distance~\cite{curry2022many}]
    Given two embeddings of the same abstract simplicial complex $K$, the distance between the Euler Characteristic Transform of each is defined as 
    $$d_{\ECT}(ECT(f(K)),ECT(g(K))) := \underset{\nu\in \mathbb{S}^{d-1}}{\int} \|\ECC_\nu(f(K)) - \ECC_\nu(g(K))\|_1 \ d\nu.$$
\end{defn}
Here, the distance is a metric, as embeddings of the same abstract simplicial complex are guaranteed to have the same Euler Characteristic when the filtrations are complete and will therefore produce a finite distance in every direction. 
Thus, all we require is the following lemma to establish that this is indeed a metric. 
\begin{lemma}
\label{IntOfMetricSd-1}
 Let $\rho$ be a metric on a space of $\mathcal{L}^p(\R)$ functions $\R \to \R$. 
 For $a,b:\mathbb{S}^{d-1} \times \R \to \R$ we denote $a_\nu(t) = a(\nu,t)$ so that $a_\nu: \S^{d-1} \to \R$. 
 Then
 $\rho'(a,b) :=\underset{\mathbb{S}^{d-1}}{\int} \rho(a_\nu,b_\nu))\,d\nu$,
 if it exists, is also a metric.
\end{lemma}

\begin{proof}
By the non-negativity of $\rho$, the integrand $\rho(a_\nu, b_\nu)$ is non-negative over the domain $\mathbb{S}^{d-1}$, so $\rho'(a, b) \geq 0$. Moreover, $\rho'(a, b) = 0$ if and only if $\rho(a_\nu, b_\nu) = 0$ for all $\nu \in \mathbb{S}^{d-1}$, which occurs if and only if $a(\nu, t) = b(\nu, t)$ for all $\nu$. Hence, $\rho'(a, b) = 0$ if and only if $a = b$.

For the symmetry property, we have:
\[
\rho'(a, b) = \int_{\mathbb{S}^{d-1}} \rho(a_\nu, b_\nu)\, d\nu = \int_{\mathbb{S}^{d-1}} \rho(b_\nu, a_\nu)\, d\nu = \rho'(b, a),
\]
where we used the symmetry of $\rho$.

To verify the triangle inequality, note that for each $\nu \in \S ^{d-1}$,
\[
\rho(a(\nu, t), c(\nu, t)) \leq \rho(a_\nu, b_\nu) + \rho(b(\nu, t), c(\nu, t)).
\]
Integrating both sides over $\S ^{d-1}$ yields:
\begin{align*}
\rho'(a, c) &= \int_{\S ^{d-1}} \rho(a_\nu, c_\nu)\, d\nu \\
&\leq \int_{\S ^{d-1}} \left[ \rho(a_\nu, b_\nu) + \rho(b_\nu, c_\nu) \right]\, d\nu \\
&= \int_{\S ^{d-1}} \rho(a_\nu, b_\nu)\, d\nu + \int_{\S ^{d-1}} \rho(b_\nu, c_\nu)\, d\nu \\
&= \rho'(a, b) + \rho'(b, c).
\end{align*}

Therefore, $\rho'$ satisfies non-negativity, identity of indiscernibles, symmetry, and the triangle inequality, so it is a metric.
\end{proof}

We similarly establish this same integral-metric property for $\R$, and thus implicitly for compact subsets of $\R$. 
\begin{comment}
    
\Oscar{Does this need a qualifier about being a finite integral? When we are talking about existence I think we meant that the integral is finite, but there is also the notion that we can actually take the integral (which I think in the case of $\S ^{d-1}$ is guaranteed by the fact that the complex is finite, so as a function of the angle the ECC distance is a sort of step function. I dont know if that makes any sense)}
\Jas{Right after the ECT definition it is mentioned that it is finite in every direction because they're different embeddings of the same complex, which guarantees the integral is finite, at least over $\S ^{d-1}$. I am getting lost though on the ECC distance as a step function of the angle so..}
\Oscar{Muncle (math uncle) indicated for the integral of a metric is a metric proof that, while making sure it is finite is important for getting a distance, that really the thing we should be worried about is whether or not we can actually take the integral. I don't know enough analysis to say this with a certainty, but I think that we are using both Riemann and Lebesgue integration, almost interchangeably (Lebesgue whenever we run into a measure 0 thing). As a result, I had some doubts about whether or not we can integrate over $S^{d-1}$, as it is unclear to me whether we are guaranteed that this function is a step function a smooth-ish function or something else.  }
\end{comment}
\begin{lemma}
    Let $\rho$ be a metric on $\mathcal{L}^p(\S ^{d-1})$ defined on functions $\S ^{d-1} \to \R$.  
    For a function $a:\S ^{d-1} \times \R \longrightarrow \R$ we write $a_t(\nu) = a(\nu,t)$ so that $a_t:\S ^{d-1} \to \R$. 
    Then
    $\rho'(a,b) :=\int_\R \rho(a_t,b_t)\, dt$, if it exists, is also a metric.
\end{lemma}

\begin{proof}
    As the proof is almost identical to that of Lemma \ref{IntOfMetricSd-1}, we omit it. 
\end{proof}

\section{Results}
In this section, we present our main stability result for the $\ECT$, showing that it is stable with respect to embedding distance. 
\begin{theorem}\label{thm:ectstablity}
Let $K$ be  an abstract simplicial complex $K$ with vertex set $V(K)$. 
Let
$$C_d = 2 \omega_{d-2} \int_0^\frac{\pi}{2}\cos(\theta)\sin^{d-2}(\theta)d\theta
\qquad \text{and }\qquad 
C_K = \max\limits_{{v \in V(K)}} |\{ \sigma \in K \mid v \in \sigma \}  |.$$ 
Then given two embeddings 
$$
f: K \rightarrow \R^d \text{ and } g:K \rightarrow \R^d$$
we have that
$$
d_{\ECT}\big{(}ECT(f(K)), ECT(g(K))\big{)} 
\leq 
2 C_K C_d \sum_{v\in V(K)}\|f(v) - g(v)\|_2.
$$
\end{theorem}

\begin{proof}
Let $h_\nu$ denote the height function in direction $\nu$, and let $h_\nu^{f(K)}$ denote the filtration of $f(K)$ induced by $h_\nu$. 
From Theorem \ref{lowerbound}, %
we have that 
\begin{equation}
\label{eqn:frompolish}
{\|\ECC_\nu(f(K)) - \ECC_\nu(g(K))\|_1 \leq 2W_{1,\infty}(\Dgm(h_\nu^{f(K)}),\Dgm(h_\nu^{g(K)})).}
\end{equation}
From Theorem \ref{upperbound}, %
we have that 
\begin{equation}%
\label{eqn:fromscary}
{\int_{\nu\in S^{d-1}}W_{1,1}(\Dgm(h_\nu^{f(K)}),\Dgm(h_\nu^{g(K)})) d\nu \leq C_K C_d \sum_{v\in V(K)}\|f(v) - g(v)\|_2}.
\end{equation} 
Since by Lemma \ref{upperbound} %
we know that for any persistence diagrams $\Dgm(K_1), \Dgm(K_2)$, 
$$W_{1,\infty}(\Dgm(K_1), \Dgm(K_2)) \leq W_{1,1}(\Dgm(K_1), \Dgm(K_2)), $$
it follows that 
$$2W_{1,\infty}(\Dgm(h_\nu^{f(K)}),\Dgm(h_\nu^{g(K)}))\leq 2W_{1,1}(\Dgm(h_\nu^{f(K)}),\Dgm(h_\nu^{g(K)})).$$
Integrating \eqref{eqn:frompolish} over $\S ^{d-1}$, and multiplying \eqref{eqn:fromscary} by two we obtain
\begin{align*}
   d_{\ECT}\big{(}\ECT(f(K)), \ECT(g(K))\big{)} &
   = \underset{\nu\in \S ^{d-1}}{\int}\|\ECC_\nu(f(K)) - \ECC_\nu(g(K))\|_1 d\nu \\
   &\leq \underset{\nu\in \S ^{d-1}}{\int}2W_{1,\infty}(\Dgm(h_\nu^{f(K)}),\Dgm(h_\nu^{g(K)})) d\nu\\
   &\leq 2\int_{\nu\in S^{d-1}}W_{1,1}(\Dgm(h_\nu^{f(K)}),\Dgm(h_\nu^{g(K)})) d\nu\\
   &\leq 2 C_K C_d \sum_{v\in V(K)}\|f(v) - g(v)\|_2.
\end{align*}
\end{proof}

\section{$\SELECT$ Stability}
In this section, we define and extend this bound to $\SELECT$ as an extension of the $\ECT$ adapted to piecewise linear functions. 
Specifically, we say that $f: \R^d \to \R$ is a in $\PL(\R^d)$ if it is continued, supported on a finite geometric simplicial complex $K$, and the values of on the faces of $K$ are determined by linear interpolation of the vertices spanning the face. 
We also denote the set of continuous functions $g: X \to Y$ by $\Func(X,Y)$.
First, we define the Lifted Euler Characteristic Transform ($\LECT$).
Note that while 
\begin{defn}[\cite{kirveslahti2024representing}]
For piecewise linear functions $f: \R ^d \rightarrow \R $, the \emph{Lifted Euler Characteristic Transform} is given by a map
\begin{equation*}
\LECT: \PL(\R ^d) \rightarrow \Func(\S ^{d-1} \times \R  \times \R, \Z )    
\end{equation*}
defined as
\begin{equation*}
\LECT(f)(\nu,a,t) = \chi \Big( \left\{x \in \R ^{d} \ \mid \ \langle x ,\nu\rangle \le a , \, f(x) = t \right\} \Big). 
\end{equation*}
\end{defn}
Note that the original definition in \cite{kirveslahti2024representing} is given on a broader class of functions, called defineable functions which are related to o-minimal structures \cite{van1998tame}, however that generalization is beyond this scope of our work. 

$\LECT$ considers only level sets of a field, which turns out to be less than ideal. 
In particular, this poses a problem when trying to define a distance on fields, the usual tool for which is integration, due to each level set being measure 0. 
In order to account for this, we instead focus on the Super Lifted Euler Characteristic Transform, $\SELECT$.

\begin{defn}[\cite{kirveslahti2024representing}]
For piecewise linear  functions $f: \R ^d \rightarrow \R $, the \emph{Super Lifted Euler Characteristic Transform} (SELECT) is given by a map
\begin{equation*}
\SELECT: \PL(\R ^d) \rightarrow \Func(\S ^{d-1} \times \R  \times \R, \Z )   
\end{equation*}
defined as
\begin{equation*}
\SELECT(f)(\nu,a,t) = \chi \big(\{x \in \R ^{d} \ \mid \ \langle x , \nu\rangle \le a , \, f(x) \geq t \} \big). 
\end{equation*}
\end{defn}

In~\cite{kirveslahti2024representing} the injectivity of $\LECT$ and  $\SELECT$ is proven which relies on having complete information in terms of the resulting functions.
The addition of this extra parameter, the value of the function, poses a challenge in terms of extending our bound to $\SELECT$. 
In addition, our bound as it is right now deals exclusively in simplicial complexes, not at all in fields. 
As such, we will now dedicate a few paragraphs to establishing the setting in which our $\SELECT$ distance works.

For our SELECT results, we consider PL functions on a simplicial complex $K$. 
These are induced by a function $\phi: V(K)\rightarrow \R_{>0}$ assigning a value to each vertex of $K$. 
We then define the extension to the full complex $K$ by $\bar \phi(\sigma) = \min_{v \in \sigma} \phi(v)$, so that $K^t = \bar \phi \inv [t,\infty)$ is a simplicial complex for every $t$.  
In particular, this means that if $\sigma' \subseteq \sigma$, $\phi(\sigma') \geq \phi(\sigma)$. 
This is analagous to the standard definition of simplex-wise monoanalogous but allows us to have superlevel sets be subcomplexes rather than sublevel sets. 

We study the functions given by two embeddings of $K$, $f, g:K \to \R^d$ for this input function $\phi:V \to \R_{>0}$ given by 
$\bar{f}(x) := \bar \phi \circ f\inv$ where we take $f\inv(x)$ to mean the simplex $\sigma$ of $K$ for which $x \in f(\sigma)\subset \R^d$.
Likewise, we set $\bar{g}:= \bar\phi \circ g^{-1}$. 
Note that $\bar \phi$ is simplex-wise monotone and therefore constant on each simplex, so for any $t$, $\bar{f} \inv(t,\infty]$ is a subcomplex of the geometric complex $f(K)$ in that if it contains a point from the image of a simplex, it contains the entirety of that simplex.
A similar result holds for $g$. 

However, the simplex-wise monotone input means that the resulting functions are piecewise constant and thus not necessarily continuous, which was a requirement for a PL function. 
If we wanted a PL function, we could have instead looked at the following construction. 
Any point $x \in \mathrm{Im}(f)$ is part of the image of some unique simplex $\sigma$ and thus has barycentric coordinates based on the vertices of that simplex, $x = \sum_{v \in \sigma} b_v f(v)$ with $b_v \in [0,1]$ and $\sum b_v = 1$. 
So then we could define $\Phi(x) = \sum_{v \in \sigma} b_v \phi(v)$ to get a PL function. 
However, since this idea is closely related to the idea of a lower star filtration (see \cite[Chapter VI.3]{Edelsbrunner2010}), we can retract $\Phi \inv[t,\infty)$ to $K^t$ for any $t$. 
Because the two sets have the same homotopy type, they have the same Euler characteristic, and so we are justified in using the piecewise constant function instead.

Thus, we define $\twiddle{f}: \S^{d-1} \times \R \times \R$ as $\SELECT(\phi\circ f^{-1})$ and $\twiddle{g}$ as $\SELECT(\phi\circ g^{-1})$. 
Now, we can take the SELECT-distance from~\cite{kirveslahti2024representing}, and modify it to reflect our restriction on the codomain of our fields.
$$
d_{\SELECT}(\twiddle{f}, \twiddle{g}) 
= \left(\int_{\R _{>0}} \int_{\S ^{d-1}}  \int_{\R } \bigm\vert \tilde{f}(\nu,a,t)-\tilde{g}(\nu,a,t) \bigm\vert^p  \ da \ d\nu \ dt\right)^{1/p}.
$$
Though this is a $p$-distance, like many others, we will restrict its use in our paper to $p=1$ in order to match the ECT distance.
We can now prove the following theorem.

\begin{theorem}
Let $K$ be an abstract simplicial complex with vertex function $\phi:V\rightarrow\R_{>0}$, and let $f,g$ be two embedding functions $f, g: K\rightarrow\R^d$. 
Define constants
$$C_d = 2 \omega_{d-2} \int_0^\frac{\pi}{2}\cos(\theta)\sin^{d-2}(\theta)d\theta, 
\qquad 
C_K = \max\limits_{{v \in V(K)}} \left|\{ \sigma \in K \mid v \in \sigma \}  \right|$$
and $r_{\max} = \max \{\phi(\sigma) \mid \sigma \in K\}$. 
Then for $\twiddle{f} = \SELECT(\phi\circ f^{-1})$ and $\twiddle{g} = \SELECT(\phi\circ g^{-1})$, we have 
$$
d_{\SELECT}(\twiddle{f}, \twiddle{g}) \leq 2 r_{\max} C_d C_K \underset{v \in V(K)}{\sum}{\|f(v) - g(v)\|_2}.
$$
\end{theorem}

\begin{proof}
First, we can rewrite our modified SELECT-distance as
\begin{align*}
d_{\SELECT}(\twiddle{f}, \twiddle{g})_1 
&= \int_{\R _{>0}} \int_{\S ^{d-1}}  \int_{\R } 
     \left|  \tilde{f}(\nu,a,t)-\tilde{g}(\nu,a,t) \right|    \ da \ d\nu \ dt \\
&= \int_{\R _{>0}} \int_{\S ^{d-1}}  \int_{\R } \left | \ECC_\nu(f(K^t),a)-\ECC_\nu(g(K^t),a) \right |  \ da \ d\nu \ dt \\
&= \int_{\R _{>0}} \int_{\S ^{d-1}} \left \| \ECC_\nu(f(K^t))-\ECC_\nu(g(K^t)) \right\|_1 \ d\nu \ dt \\
&= \int_{\R _{>0}} d_{\ECT}\big(ECT(f(K^t)), ECT(g(K^t))\big) \ dt.
\end{align*}
We know by our ECT stability theorem, Thm.~\ref{thm:ectstablity}, that 
$$
d_{\ECT}\big(ECT(f(K^t)), ECT(g(K^t))\big) 
\leq 2 C_d C_{K^t} \sum_{v\in V(K^t)}\|f(v) - g(v)\|_2,
$$
so we obtain
\begin{align*}
d_{\SELECT}(\twiddle{f}, \twiddle{g})_1 &\leq \int_{\R _{>0}} 2 C_d C_{K^t} \sum_{v\in V(K^t)}\|f(v) - g(v)\|_2 \ dt \\
&\leq 2 C_d \int_{\R _{>0}} C_{K^t} \sum_{v\in V(K^t)}\|f(v) - g(v)\|_2 \ dt.
\end{align*}
We are able to pull $C_d$ out of the integral because it is not dependent on the subcomplex. The other terms, however, are dependent on the subcomplex. 
The constant
$$
C_{K^t} = \underset{v \in V(K^t)}{\max} |\{\sigma \in K^t \mid v \in \sigma\}|
$$
is the maximum number of cofaces for any vertex in $K^t$. 
As $t$ increases and the subcomplex contracts, the number of cofaces for each vertex will only ever decrease or stay the same. 
This means that $C_{K^t} \leq C_{K}$ for all $t \in \R _{>0}$.
Furthermore, for 
$t > r_{\max}=\underset{\sigma\in K}{\max} \{\phi(\sigma)\}$, $K^t = \emptyset$, so for large enough $t$, $C_{K^t} = 0$. 
This means that we can restrict our bounds of integration, loosen our upper bound, and then pull out $C_K$, as follows:
\begin{align*}
d_{\SELECT}(\twiddle{f}, \twiddle{g})_1 &\leq 2 C_d \int_{\R _{>0}} C_{K^t} \sum_{v\in V(K^t)}\|f(v) - g(v)\|_2 \ dt \\
&\leq 2 C_d \int_{0}^{r_{\max}} C_{K^t} \sum_{v\in V(K^t)}\|f(v) - g(v)\|_2 \ dt \\
&\leq 2 C_d \int_{0}^{r_{\max}} C_{K} \sum_{v\in V(K^t)}\|f(v) - g(v)\|_2 \ dt \\
&\leq 2 C_d C_K \int_{0}^{r_{\max}} \sum_{v\in V(K^t)}\|f(v) - g(v)\|_2  \ dt.
\end{align*}
We can use similar logic to deal with the remaining integrand. 
$V(K^t) \subseteq V(K)$ for every $t$, so 
$$
\sum_{v\in V(K^t)} \|f(v) - g(v)\|_2 
\leq 
\sum_{v\in V(K)} \|f(v) - g(v)\|_2
$$
for every $t$. 
This leads to our final upper bound.
\begin{align*}
d_{\SELECT}(\twiddle{f}, \twiddle{g})_1 &\leq 2 C_d C_K \int_{0}^{r_{\max}} \sum_{v\in V(K^t)}\|f(v) - g(v)\|_2  \ dt \\ 
&\leq 2 C_d C_K \int_{0}^{r_{\max}} \sum_{v \in V(K)}\|f(v) - g(v)\|_2  \ dt \\
&\leq 2 C_d C_K \sum_{v\in V(K)}\|f(v) - g(v)\|_2 \int_{0}^{r_{\max}} 1 \ dt \\
&\leq 2 r_{\max} C_d C_K \sum_{v\in V(K)}\|f(v) - g(v)\|_2.
\end{align*}
\end{proof}

\section{Discussion and Conclusion} %

The stability of the Euler Characteristic Transform (ECT) is crucial for its dependable application in shape analysis. 
This work contributes to this area by establishing an explicit upper bound (Theorem \ref{thm:ectstablity}) on the $d_{\ECT}$ distance. 
This bound arises when comparing two distinct geometric embeddings of the same abstract simplicial complex $K$ and directly links variations in the ECT to the sum of $L_2$ distances between corresponding vertex positions. 
This quantifies ECT's sensitivity to geometric perturbations when the underlying combinatorial topology is fixed. 
We subsequently extended this stability framework to the Super Lifted Euler Characteristic Transform (SELECT), considering fields induced by simplex-wise monotone functions $\phi$ on $K$, and similarly connected the $d_{\SELECT}$ stability to vertex displacements and the range of $\phi$.

These findings provide important theoretical guarantees for the robustness of ECT and SELECT. Such stability is particularly relevant in applications involving deformable objects or variations in pose where the intrinsic mesh topology is preserved, supporting the reliable use of these TDA descriptors in analytical pipelines.

However, our results are contingent on specific conditions: primarily, the comparison of embeddings of an identical abstract complex and, for SELECT, the use of a particular class of field-inducing functions. 
A significant direction for future research is therefore the generalization of these stability results. 
This includes extending SELECT stability to more commonly used field representations, such as those derived from piecewise-linear interpolation, and developing frameworks for ECT/SELECT comparisons across different underlying combinatorial structures. 
Refining the derived bounds also remains an important objective.

\section{Acknowledgements}
This paper is the result of work from a team of undergraduate researchers in Summer 2024 supported by the SURIEM and ACRES REU programs at Michigan State University. 
The authors would like to acknowledge the gracious support of the National Science Foundation through Award 
Nos.~DMS-2244461, %
CCF-2349002, %
CCF-2142713, %
and CCF-2106578; as well as Michigan State University.
We would also like to thank Dr.~Robert Bell; and Dr.~Mahmoud Parvizi and Dr.~Alex Dickson for organizing the SURIEM and ACRES REUs respectively. 

  \bibliographystyle{plain}
   \bibliography{references}

\end{document}